\documentclass[runningheads]{llncs}

\usepackage{amsmath,amssymb,amsfonts}
\usepackage{microtype}
\usepackage{booktabs}
\usepackage{xcolor}
\usepackage[hidelinks]{hyperref}

\usepackage[ruled,vlined]{algorithm2e}
\SetKwProg{Fn}{Function}{}{end}
\SetKwFunction{FRecurs}{FnRecursive}

\usepackage{wrapfig}

\usepackage{graphicx}
\usepackage{caption}
\usepackage{subcaption}
\usepackage{tikz}
\usetikzlibrary{arrows.meta}

\newcommand{\inc}[1]{\textcolor{blue}{#1}}

\title{Incremental Neural Network Verification via Learned Conflicts}
\author{
Raya Elsaleh\inst{1} \and
Liam Davis\inst{2} \and
Haoze Wu\inst{2} \and
Guy Katz\inst{1}
}

\institute{
Hebrew University of Jerusalem, Jerusalem, Israel \\
\email{\{raya.elsaleh,g.katz\}@mail.huji.ac.il}
\and
Amherst College, Amherst, USA \\
\email{\{ljdavis27,hwu\}@amherst.edu}
}

\begin{document}
\maketitle

\begin{abstract}
Neural network verification is often used as a core component within larger
analysis procedures, which generate sequences of closely related verification
queries over the same network.
In existing neural network verifiers, each query is typically solved
independently, and information learned during previous runs is discarded, leading
to repeated exploration of the same infeasible regions of the search space.
In this work, we aim to expedite verification by reducing this redundancy.
We propose an incremental verification technique that reuses learned conflicts
across related verification queries. The technique can be added on top of any 
branch-and-bound-based neural network verifier.
During verification, the verifier records conflicts corresponding to learned 
infeasible combinations of activation phases, and retains them across runs.
We formalize a refinement relation between verification queries and show that
conflicts learned for a query remain valid under refinement, enabling sound
conflict inheritance.
Inherited conflicts are handled using a SAT solver to perform consistency checks
and propagation, allowing infeasible subproblems to be detected and pruned early
during search.
We implement the proposed technique in the Marabou verifier and evaluate it on three verification
tasks: local robustness radius determination, verification with input splitting,
and minimal sufficient feature set extraction.
Our experiments show that incremental conflict reuse reduces verification effort
and yields speedups of up to 1.9$\times$ over a non-incremental baseline.
\end{abstract}

\section{Introduction}

Deep neural networks (DNNs) are increasingly deployed in safety-critical
applications~\cite{AmOlStChScMa16}, including autonomous driving~\cite{autonomous_driving,PaGoYu19},
medical diagnosis~\cite{medical_ai}, and aerospace
systems~\cite{robustness_airbus,reluplex}. Their strong empirical performance has driven
major advances in tasks such as control and image
recognition. Despite this success, neural networks remain
largely opaque and difficult to reason about, raising serious concerns
regarding their reliability and safety in critical settings.

To address this issue, numerous approaches have been developed for neural
network verification, with the goal of providing rigorous guarantees about network
behavior. For networks with piecewise-linear activation functions,
such as ReLUs, the verification problem is NP-complete~\cite{reluplex},
which hinders scalability; but as neural networks are adopted across
increasingly diverse and high-stakes domains, the demand for scalable
verification techniques continues to grow.  Improving the efficiency
of neural network verification is therefore imperative.

Examining how DNN verification is used in practice, we observe that it
is often not performed as a
single, isolated query, but rather invoked repeatedly within larger analysis
procedures.
For example, in formal explainability, verification queries are issued repeatedly
to reason about the contribution of different input features to a given
prediction under progressively refined 
constraints~\cite{towards_xai,verix_plus,VeriX,BoElDuBaKa26}; similarly, in robustness
radius computation, verification queries are invoked iteratively to narrow down
the maximum safe perturbation radius around a given input~\cite{robustness_airbus,robustness,StWuZeJuKaBaKo23}.
Such analyses naturally give rise to sequences of closely related verification
queries that differ in limited aspects of their specifications, such as
refined input domains or strengthened output constraints.
Despite this, current verification tools do not explicitly exploit this structural
similarity:
each query is restarted from
scratch, and information derived during previous runs is discarded.

In this work, we seek to mitigate these inefficiencies
by reusing lemmas derived from earlier verification runs to 
accelerate verification of subsequent queries. 
Similar incremental solving techniques have proven successful in SAT and SMT
solving~\cite{smt_handbook,inc_smt,incremental_sat_1}. In the case of neural network
verification, previous work has considered proof transfer for abstract-interpretation-based
methods~\cite{proof_transfer} and warm-starting 
branch-and-bound by heuristically resuming the search from leaf nodes of search trees generated 
from prior solver runs~\cite{ivan,ivan_2}. However, none of the previous work has explored 
reusing lemmas across multiple invocations of branch-and-bound-based complete verifiers, 
which is the focus of this work.

Our incremental verification approach is designed for sequences of closely
related properties on the same neural network.
Our approach records conflicts that arise during branch-and-bound verification,
where each conflict captures an infeasible combination of branching decisions.
These conflicts are preserved beyond individual verification runs and reused
in subsequent queries with refined specifications. We formally define a refinement
condition for the conflict to remain valid and show that this condition can be
established in several important applications of neural network verification.
The inherited conflicts allow the verifier to immediately prune previously
explored infeasible regions, avoiding redundant analysis and computation.

We develop a framework for conflict recording and sound reuse that integrates directly
with branch-and-bound-based verifiers, and employ a SAT solver to efficiently
manage and apply large collections of learned conflicts during solving.
The technique can be added on top of any branch-and-bound-based neural network
verifier. To investigate the effectiveness of the proposed approach,
we instantiate it in the Marabou verifier~\cite{marabou,reluplex,marabou2}
and perform a thorough evaluation on three representative verification
tasks---robustness radius computation, iterative input splitting, and formal 
explanation---and demonstrate consistent empirical reductions in verification runtime, with speedups of up to
$1.9\times$ compared to the non-incremental baseline.

The rest of the paper is organized as follows.
Section~\ref{sec:background} introduces the necessary background on neural network
verification, branch-and-bound search, case splitting, and bound propagation.
Section~\ref{sec:incremental-verification} presents our incremental verification
framework, including conflict clauses, query refinement, and the sound reuse of
learned conflicts, and describes the implementation details of the proposed
approach.
Section~\ref{sec:verification-tasks} evaluates the approach on several
verification tasks and discusses the experimental results.
Section~\ref{sec:related-work} reviews related work.
Finally, Section~\ref{sec:future-work} outlines directions for future work, and
Section~\ref{sec:conclusion} concludes the paper.

\section{Background}
\label{sec:background}

\subsection{Neural Networks and the Verification Problem}
\paragraph{Neural Networks.}
A deep neural network is a sequence of $L$ layers, where layer $i$ contains
$d^{(i)}$ neurons, with $i=0$ denoting the input layer and $i=L$ the output
layer~\cite{DeepLearning}. Such a network defines a function
$f:\mathbb{R}^{d^{(0)}}\to\mathbb{R}^{d^{(L)}}$.
Each layer applies an affine transformation followed by an activation function.
In this work, we focus on networks with ReLU activations,
$\mathrm{ReLU}(t)=\max(0,t)$, though the approach extends naturally to other
piecewise-linear activations.

For each layer $i$, the computation is given by
$z^{(i)}=\mathbf{W}^{(i)}x^{(i-1)}+\mathbf{b}^{(i)}$ and
$x^{(i)}=\mathrm{ReLU}(z^{(i)})$, where $x^{(0)}=x_0$ is the network input.
We denote by $z^{(i)}_j$ and $x^{(i)}_j$ the pre- and post-activation values of the
$j$-th neuron in layer $i$.
Each ReLU neuron is associated with a Boolean phase variable $r^{(i)}_j$,
where $r^{(i)}_j=\top$ corresponds to the active phase ($z^{(i)}_j\ge 0$) and
$r^{(i)}_j=\bot$ to the inactive phase ($z^{(i)}_j<0$).

\paragraph{Verification Queries.}
Neural network verification asks whether a network can exhibit undesired
behavior.
A verification query for a network
$f:\mathbb{R}^{d^{(0)}}\to\mathbb{R}^{d^{(L)}}$ is defined by a pair
$(\mathcal{X},\mathcal{Y})$, where
$\mathcal{X}\subseteq\mathbb{R}^{d^{(0)}}$ and
$\mathcal{Y}\subseteq\mathbb{R}^{d^{(L)}}$ specify the input and output regions,
respectively, and asks whether $\exists x\in\mathcal{X}$ such that $f(x)\in\mathcal{Y}$.
Typically, both regions are described by linear constraints, with
$\mathcal{Y}$ representing an undesired output set.
A query is answered \textsf{SAT} if such an input exists and \textsf{UNSAT}
otherwise.

\subsection{Branch-and-Bound Verification}
Most modern neural network verifiers combine case splitting with bound
propagation, a paradigm commonly referred to as branch-and-bound, to reason about
the disjunctive behavior induced by ReLU activations%
~\cite{reluplex,marabou,crown1,crown2,dependency_1,neuralsat,nnenum,nnv}.
In this work, we focus on the case-splitting component.

\paragraph{Case Splitting.}
Verification of ReLU neural networks relies on case splitting over ReLU activation
phases.
Each split fixes a single phase variable $r^{(i)}_j$ to either the active or
inactive case, i.e., $r^{(i)}_j=\top$ or $r^{(i)}_j=\bot$, thereby refining the
original verification query into independent subproblems.

\paragraph{Branch-and-Bound Search and Search Tree.}
Given a neural network $f$ and a verification query
$q=(\mathcal{X},\mathcal{Y})$, branch-and-bound verification explores the space of
ReLU activation phases induced by $q$ by repeated case splitting, inducing a
search tree (Figure~\ref{fig:bab-tree}).

The root of the tree corresponds to the original query with no fixed phase decisions.
Each node is associated with a \emph{partial assignment}, also called a \emph{decision
trail}, $\pi=\{\ell_1,\ldots,\ell_k\}$, where each literal
$\ell_i\in\{r^{(i)}_j,\neg r^{(i)}_j\}$ fixes the activation phase of a ReLU
neuron.
The decision trail $\pi$ defines a refined subproblem obtained by conjoining the
corresponding phase constraints with the original query, and edges correspond to
individual case splits extending $\pi$ by one additional literal.
For readability, Figure~\ref{fig:bab-tree} uses $r_1,r_2,r_3,r_4$ to denote phase variables. 
For example, the double-circled node in Figure~\ref{fig:bab-tree} corresponds to
the decision trail $\pi=\{\neg r_1,\, r_2\}$.
At each node, the verifier applies bound propagation to compute
over-approximations of neuron pre- and post-activation values.
If these bounds imply that no input consistent with $\pi$ can satisfy the query,
then $\pi$ is infeasible and the node is declared \textsf{UNSAT} and its subtree is pruned.
If a concrete input $x\in\mathcal{X}$ consistent with $\pi$ satisfies
$f(x)\in\mathcal{Y}$, the node is declared \textsf{SAT} (e.g., the
$\pi=\{r_1,\neg r_2\}$ leaf in Figure~\ref{fig:bab-tree}), after which the search
terminates and unexplored branches need not be visited.

The search terminates when either a \textsf{SAT} leaf is found or all explored
leaves are \textsf{UNSAT}.
In the latter case, no violating input exists and the property is verified.
This branch-and-bound paradigm underlies modern neural network 
verifiers~\cite{marabou,reluplex,crown1,crown2,dependency_1,nnv}.
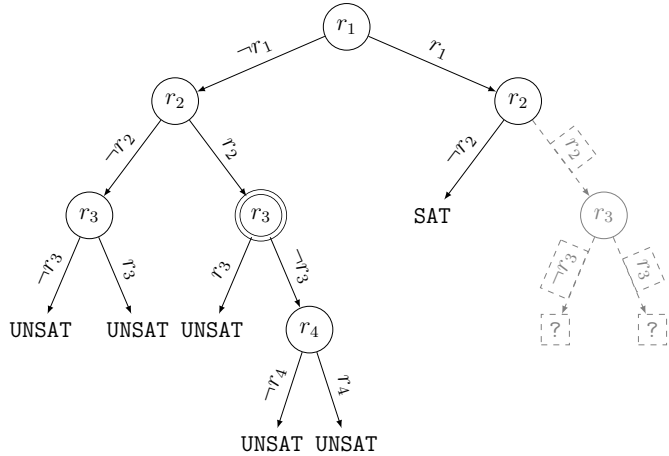
\begin{figure}[t]
\centering
{\large
\scalebox{0.76}{
\begin{tikzpicture}[
  level 1/.style={sibling distance=6cm, level distance=1.3cm},
  level 2/.style={sibling distance=3cm, level distance=2cm},
  level 3/.style={sibling distance=1.7cm, level distance=2cm},
  level 4/.style={sibling distance=1.3cm, level distance=2cm},
  node/.style={circle, draw, minimum size=.6cm},
  edge/.style={->, >=latex},
  pruned/.style={draw=gray, text=gray},
  prunedEdge/.style={->, >=latex, draw=gray, dashed},
]
\node[node] (a) {$r_1$}
child { node[node] (b) {$r_2$}
  child { node[node] (c) {$r_3$}
    child { node[draw=none] {\texttt{UNSAT}}
      edge from parent[edge] node[sloped, above] {$\neg r_3$} }
    child { node[draw=none] {\texttt{UNSAT}}
      edge from parent[edge] node[sloped, above] {$r_3$} }
    edge from parent[edge] node[sloped, above] {$\neg r_2$} }
  child { node[node, double, double distance=2pt] (d) {$r_3$}
    child { node[draw=none] {\texttt{UNSAT}}
      edge from parent[edge] node[sloped, above] {$r_3$} }
    child { node[node] {$r_4$}
      child { node[draw=none] {\texttt{UNSAT}}
        edge from parent[edge] node[sloped, above] {$\neg r_4$} }
      child { node[draw=none] {\texttt{UNSAT}}
        edge from parent[edge] node[sloped, above] {$r_4$} }
      edge from parent[edge] node[sloped, above] {$\neg r_3$} }
      edge from parent[edge] node[sloped, above] {$r_2$} }
  edge from parent[edge] node[sloped, above] {$\neg r_1$} }
child { node[node] (e) {$r_2$}
  child { node[draw=none] (satleaf) {\texttt{SAT}}
    edge from parent[edge] node[sloped, above] {$\neg r_2$} }
  child { node[node, pruned] (g) {$r_3$}
    child { node[draw=none, pruned] {\texttt{?}}
      edge from parent[prunedEdge] node[sloped, above, pruned] {$\neg r_3$} }
    child { node[draw=none, pruned] {\texttt{?}}
      edge from parent[prunedEdge] node[sloped, above, pruned] {$r_3$} }
    edge from parent[prunedEdge] node[sloped, above, pruned] {$r_2$} }
  edge from parent[edge] node[sloped, above] {$r_1$} }
;
\end{tikzpicture}
}}
\caption{Branch-and-bound search tree over ReLU phase decisions.}
\label{fig:bab-tree}
\end{figure}

\paragraph{Bound Propagation.}
Bound propagation analyzes the feasibility of subproblems arising during
branch-and-bound search (i.e., nodes in the search tree) by computing
over-approximations of neuron values under a partial assignment $\pi$%
~\cite{deeppoly,crown4bp,reluval}.
These bounds detect infeasible assignments and implied activation phases,
enabling early pruning of the search tree
(e.g., the \textsf{UNSAT} leaves in Figure~\ref{fig:bab-tree}).
% --------------------------------------------------

\section{Incremental Verification via Learned Conflicts}
\label{sec:incremental-verification}

\subsection{Conflicts and Query Refinement}
\paragraph{Learned Conflicts.}
During branch-and-bound verification, infeasible subproblems encountered along
the search can be recorded and reused in the form of conflict clauses.

\begin{definition}[Conflict Clause]
\label{def:conflict-clause}
A conflict clause for a verification query $q$ is a set of literals
$c = \{\ell_1, \ldots, \ell_k\}$ such that the corresponding CNF clause
\(
(\neg \ell_1 \lor \neg \ell_2 \lor \cdots \lor \neg \ell_k)
\)
is logically implied by $q$.
\end{definition}

When a decision trail $\pi = \{\ell_1, \ldots, \ell_k\}$ is found to be \textsf{UNSAT},
the infeasible combination can be summarized by the conflict clause
$(\neg \ell_1 \lor \cdots \lor \neg \ell_k)$.

\paragraph{Query Refinement.}
To enable sound reuse of conflict clauses across verification queries,
we formalize when one query is a \emph{refinement} of another.

\begin{definition}[Query Refinement]
\label{def:refinement}
A verification query $q_2$ is a \emph{refinement} of another query $q_1$, denoted
$q_2 \preceq q_1$, if both queries are defined over the same network $f$ and
\[
\mathcal{X}(q_2) \subseteq \mathcal{X}(q_1)
\quad \text{and} \quad
\mathcal{Y}(q_2) \subseteq \mathcal{Y}(q_1).
\]
\end{definition}

Intuitively, $q_2$ imposes stricter constraints than $q_1$, and therefore admits
a smaller feasible region. In particular, any conjunction of constraints that is
infeasible under $q_1$ remains infeasible under $q_2$.
We formalize this monotonicity of infeasibility in
Lemma~\ref{lem:monotonicity}.

In practice, many verification workflows progressively strengthen the input
domain while leaving the output constraints unchanged, as in the use cases
considered below. Nevertheless, our framework supports refinements of both input
and output domains.

\subsection{Soundness of Conflict Reuse}

\begin{lemma}[Monotonicity of Infeasibility Under Refinement]
\label{lem:monotonicity}
Let $q_1$ and $q_2$ be verification queries such that $q_2 \preceq q_1$, and let
$\{\ell_1,\ldots,\ell_k\}$ be a set of literals.
If the subproblem
\[
q_1 \land \ell_1 \land \cdots \land \ell_k
\]
is infeasible, then the subproblem
\[
q_2 \land \ell_1 \land \cdots \land \ell_k
\]
is also infeasible.
\end{lemma}

\noindent
Lemma~\ref{lem:monotonicity} formalizes the monotonicity of infeasibility under
query refinement. The proof is provided in
Appendix~\ref{app:conflict-soundness}.

\begin{theorem}[Sound Conflict Reuse under Refinement]
\label{thm:conflict-reuse}
Let $q_1$ and $q_2$ be verification queries such that $q_2 \preceq q_1$.
Let $c=\{\ell_1,\ldots,\ell_k\}$ be a conflict clause for $q_1$.
Then $c$ is also a conflict clause for $q_2$.
\end{theorem}

\noindent
Theorem~\ref{thm:conflict-reuse} follows directly from
Lemma~\ref{lem:monotonicity} and establishes the soundness of conflict inheritance
across refined queries. A detailed proof is given in
Appendix~\ref{app:conflict-soundness}.

As a result, conflicts learned for one query can be reused to prune infeasible
regions in refined queries without re-exploration.

\subsection{Using Conflicts during Verification}

Having established the soundness of conflict reuse under query refinement, we now
describe how conflict clauses are used during verification to prune the search
space.
Specifically, we employ a SAT solver to reason about inherited conflict clauses
during branch-and-bound search as an additional pruning and propagation step.

\paragraph{Checking Consistency via SAT.}
At the start of each verification query, the verifier inherits a set of conflict
clauses $\mathcal{C}$ learned from previous related queries.
These clauses are encoded once as CNF clauses in a SAT solver.

At each node of the branch-and-bound search tree, after standard bound propagation
is applied, the verifier invokes the SAT solver to check the consistency of the
current partial assignment with the inherited conflict clauses $\mathcal{C}$.
Let $\alpha$ denote the current partial assignment to ReLU phase variables.
Each literal $\ell \in \alpha$ is asserted as a unit assumption in the SAT solver.

\paragraph{Soundness of SAT-Based Pruning.}
SAT-based reasoning under assumptions can yield two outcomes: either the
current trail 
is \textsf{UNSAT}, certifying the subproblem as infeasible, or unit
propagation derives implied assignments that further restrict the search space.
An \textsf{UNSAT} result precludes any extension of the current partial
assignment, while implied assignments enforce necessary literals without
excluding feasible solutions.

\begin{lemma}[Soundness of SAT-Based Pruning]
\label{lem:sat-pruning}
Let $\alpha$ be a partial assignment and $\mathcal{C}$ a set of conflict clauses
implied by a verification query $q$.
If the propositional CNF formula constructed from $\alpha$ and $\mathcal{C}$ is
\textsf{UNSAT}, then no extension of $\alpha$ can correspond to a feasible
solution for $q$.
\end{lemma}

\begin{lemma}[Soundness of SAT-Based Implied Assignments]
\label{lem:sat-implied}
Let $\alpha$ be a partial assignment and $\mathcal{C}$ a set of conflict clauses
implied by a verification query $q$.
Let $\alpha'$ be the assignment obtained by SAT-based reasoning over $\alpha$ and
$\mathcal{C}$, and let $\mathcal{L} = \alpha' \setminus \alpha$ denote the implied
assignments.
Then any feasible solution for $q$ extending $\alpha$ must satisfy all literals
in $\mathcal{L}$.
\end{lemma}

Together, Lemmas~\ref{lem:sat-pruning} and~\ref{lem:sat-implied} establish that
SAT-based reasoning over inherited conflict clauses enables both sound pruning
and sound propagation, reducing the search space while preserving correctness.
Full proofs are provided in Appendix~\ref{app:sat-soundness}.

\subsection{Incremental Verification Workflow and Integration}
This subsection describes the end-to-end workflow by which incremental conflict
reuse is integrated into branch-and-bound verification.
At a high level, verification proceeds over a sequence of related queries, where
conflicts learned during earlier runs are recorded and selectively inherited by
subsequent queries.
Algorithm~\ref{alg:ica} presents the Incremental Conflict Analyser component that
manages conflict storage and SAT-based reasoning, while
Algorithm~\ref{alg:bab-inc} shows how the component is invoked within the
branch-and-bound search loop.

\begin{algorithm}[t]
\caption{\textsc{ICA}: Incremental Conflict Analyser Component}
\label{alg:ica}
\DontPrintSemicolon
\LinesNumbered

\textbf{Fields:}\\
\qquad $\textsc{Pool}$: map $id \mapsto$ set of conflicts $\mathcal{C}_{id}$\;
\qquad $\textsc{Sat}$: SAT solver instance\;

\BlankLine
\Fn{\textsc{BeginQuery}$(\mathcal{I})$}{
    
    \nl
    $\textsc{Sat}.\textsc{Reset}()$\tcp*{fresh SAT instance for this query} \label{ln:ica-begin}

     \ForEach{$id' \in \mathcal{I}$}{
        \ForEach{$c \in \textsc{Pool}[id']$}{
            $\textsc{Sat}.\textsc{AddClauseAsCNF}(c)$\;
        }
    } \label{ln:ica-init-clauses}
}

\BlankLine
\Fn{\textsc{Propagate}$(Bounds)$}{
    \nl
    $\alpha \gets \textsc{ExtractPartialAssignment}(Bounds)$\;
    \label{ln:ica-prop-start}
    \nl
    $res \gets \textsc{Sat}.\textsc{SolveUnderAssumptions}(\alpha)$\;

    \nl
    \uIf{$res=\textsf{UNSAT}$}{
        \nl
        \Return{$(\textsf{UNSAT}, \emptyset)$}\;
    }
    \nl
    $\Delta_{\textsf{sat}} \gets \textsc{Sat}.\textsc{GetUnitImpliedLiterals}()$\;
    \nl
    $Bounds.\textsc{ApplyImpliedLiterals}(\Delta_{\textsf{sat}})$\;

    \nl
    \Return{$(\textsf{SAT}, Bounds)$}\;
    \label{ln:ica-prop-end}
}

\BlankLine
\Fn{\textsc{RecordConflict}$(id,c)$}{
    \nl
    \If{$\exists c' \in \textsc{Pool}[id] \text{ such that } c' \subseteq c$}{
        \label{ln:ica-record-start}
        \nl
        \Return{}\;
    }

    \nl
    $\textsc{Pool}[id] \gets \textsc{Pool}[id] \cup \{c\}$\;
    \nl
    $\textsc{Sat}.\textsc{AddClauseAsCNF}(c)$\;
    \label{ln:ica-record-end}
}
\end{algorithm}

\paragraph{Incremental Conflict Analyser.}
\label{sec:ica}
The core component enabling incremental conflict reuse is the
\emph{Incremental Conflict Analyser} (ICA), shown in
Algorithm~\ref{alg:ica}.
The ICA is responsible for storing, retrieving, and applying learned conflict
clauses across related verification queries.

Each verification query $q$ is issued together with an \emph{inheritance set}
$\mathcal{I}$ of query identifiers whose learned conflicts may be reused.
The ICA maintains a global pool of conflict clauses indexed by query identifier,
along with a SAT solver instance populated with the conflicts active for the
current query.

At the start of each verification, \textsc{ICA.BeginQuery} is invoked
(Lines~\ref{ln:ica-begin}--\ref{ln:ica-init-clauses}) to reset the SAT solver and
load all conflict clauses associated with the identifiers in $\mathcal{I}$.
During branch-and-bound search, \textsc{ICA.Propagate} performs SAT-based reasoning
over the current subproblem by passing the current partial assignment over ReLU
phase variables as assumptions to the SAT solver
(Lines~\ref{ln:ica-prop-start}--\ref{ln:ica-prop-end}).
If the SAT instance is \textsf{UNSAT}, this result is reported back to the
verifier for pruning; otherwise, implied ReLU phase assignments are returned and
applied as additional bounds.
Finally, \textsc{ICA.RecordConflict} records newly discovered conflicts
(Lines~\ref{ln:ica-record-start}--\ref{ln:ica-record-end}), making them
available for reuse in subsequent refined queries.

\begin{algorithm}[t]
\caption{Branch-and-Bound Verification with \inc{Incremental Conflict Reuse}}
\label{alg:bab-inc}
\DontPrintSemicolon

\KwIn{
Verification query $q$ on network $f$;\\
\inc{query identifier $id$; inherited identifier set $\mathcal{I}$;\\
Incremental Conflict Analyser object \textsc{ICA}}.
}
\KwOut{$\textsf{SAT}$, $\textsf{UNSAT}$, or $\textsf{TIMEOUT}$.}

\BlankLine
\nl
\inc{\textsc{ICA}.\textsc{BeginQuery}$(\mathcal{I})$}\; \label{ln:begin-query}
\nl
$\pi \gets \emptyset$\tcp*{Current decision trail}
\nl
$Q \gets$ stack containing $\pi$\tcp*{DFS (or priority queue for best-first)}

\BlankLine
\nl
\While{$Q \neq \emptyset$}{
  \nl
  $\pi \gets Q.\textsc{Pop}()$\;
  \nl
  $(status,\; Bounds) \gets \textsc{Propagate}(q,\pi)$
  \tcp*{Standard bound propagation} \label{ln:numeric-prop}

  \nl
  \uIf{$status=\textsf{SAT}$}{
    \nl
    \Return{$\textsf{SAT}$}\tcp*{Counterexample found}
  }

  \nl
  \uIf{$status=\textsf{UNSAT}$}{
    \inc{\tcp{Record conflict}}
    \nl
    \inc{$c \gets \textsc{ExtractConflict}(\pi)$}\;
    \nl
    \inc{\textsc{ICA}.\textsc{RecordConflict}$(id, c)$}\;
    \label{ln:bab-record-conflict}
    \nl
    \textbf{continue}\;
  }

  \nl
  \Else(\tcp*[f]{\textsf{UNKNOWN}}){
    \inc{\tcp{Incremental conflict reasoning}}
    \nl
    $\inc{(icaStatus,\; Bounds) \gets
    \textsc{ICA}.\textsc{Propagate}(Bounds)}$\;
    \label{ln:sat-prop}

    \nl
    \uIf(\inc{\tcp*[f]{Violates inherited conflict}}){\inc{$icaStatus=\textsf{UNSAT}$}}{
      \nl
      \inc{\textbf{continue}}\;
    }
  }

  \nl
  $r \gets \textsc{ChooseSplit}(Bounds)$
  \label{ln:bab-split}
  \tcp*{Split on an undecided ReLU phase}

  \nl
  $Q.\textsc{Push}(\pi \cup \{r\})$\;
  \nl
  $Q.\textsc{Push}(\pi \cup \{\neg r\})$\;
}
\nl
\Return{$\textsf{UNSAT}$}
\end{algorithm}

\paragraph{Branch-and-Bound with Incremental Reuse.}
Algorithm~\ref{alg:bab-inc} presents a standard branch-and-bound verification
procedure augmented with incremental conflict reuse; incremental extensions are
highlighted in \inc{blue}.
At the start of verification, the Incremental Conflict Analyser is initialized
for the current query by invoking
$\textsc{ICA}.\textsc{BeginQuery}(\mathcal{I})$
(Line~\ref{ln:begin-query}), which activates the conflicts inherited from prior
queries.

The main search loop follows the standard branch-and-bound structure.
For each decision trail $\pi$, numeric bound propagation is applied first
(Line~\ref{ln:numeric-prop}).
If a concrete counterexample is found, the procedure terminates with
\textsf{SAT}.
If the subproblem is proven infeasible, a conflict is extracted from the current
trail and recorded via
$\textsc{ICA}.\textsc{RecordConflict}$, after which the branch is pruned.

If numeric propagation is inconclusive, the verifier invokes incremental conflict
reasoning by calling
$\textsc{ICA}.\textsc{Propagate}$ (Line~\ref{ln:sat-prop}).
This step checks the current partial assignment against inherited conflict
clauses.
If the SAT solver reports \textsf{UNSAT}, the node is immediately pruned;
otherwise, any implied ReLU phase assignments are applied as additional
propagation step. The solver then selects an undecided ReLU phase and 
branches using the existing branching heuristic.

The proposed incremental conflict reuse mechanism can be integrated
into branch-and-bound verification as a lightweight, sound extension that 
preserves the solver's core reasoning while reducing redundant exploration.

% ------------------------------------------------------------------
\section{Incremental Verification Use Cases}
\label{sec:verification-tasks}

In this section, we consider three representative use cases of incremental verification
and evaluate the effectiveness of the proposed conflict-reuse mechanism.
For each use case, we explain how its
structure gives rise to related queries, formally establish the refinement
relations, and evaluate the effectiveness of incremental conflict reuse
by comparing the performance of the incremental approach against a non-incremental baseline.
Before describing the use cases, we first summarize our implementation as well as the 
evaluation metrics considered in our experiments.

\subsection*{Implementation Details}
Our implementation integrates the Incremental Conflict Analyser into the
Marabou verifier~\cite{marabou,marabou2}, using the CaDiCaL SAT solver~\cite{cadical}
for conflict reasoning.
All verification queries induced by the same task share a single ICA instance.

\paragraph{Propagation from Inherited Conflicts.}
Inherited conflict clauses influence verification in two ways.
SAT-based reasoning may detect that the current partial assignment is infeasible,
allowing the verifier to prune the subproblem immediately, or it may derive
additional ReLU phase assignments via unit propagation.
These implied assignments introduce additional linear constraints that are
integrated into the verifier's existing numeric reasoning, further restricting
the feasible region.
To quantify the impact of conflict inheritance, we measure the total number of
such effects---both pruned subproblems and implied assignments---over the course of a
verification task.

\medskip
We now turn to the individual verification tasks considered in our evaluation.

\subsection{Use Case 1: Determining the Local Robustness Radius}
\label{sec:robustness-verification}

In this common use case, the goal is to identify the largest neighborhood around a
reference input within which a neural network's output remains consistent.
Given a reference input $x_0$ and a task-specific notion of output consistency
(e.g., preservation of the predicted class or bounded deviation of the output),
the goal is to determine tight lower and upper bounds on the maximal robustness
radius, up to a specified precision.

We formalize this task via \emph{local robustness verification queries} and the
associated \emph{local robustness radius}.

\paragraph{Local robustness verification queries given $\varepsilon$.}
In the classification setting considered here, a robustness query asks whether
the predicted class changes within a bounded perturbation region around a
reference input $x_0$.
Given a norm $\|\cdot\|_p$ and a radius $\varepsilon > 0$, the network is
\emph{not locally robust} at $x_0$ with respect to $\varepsilon$ if there exists 
an input $x$ within the
$\varepsilon$-ball around $x_0$ that violates the desired output property.
Formally, let $\mathcal{P}(x_0, x)$ denote a task-specific predicate capturing
misclassification, the network is not locally robust if
\[
\exists x \in \mathbb{R}^{n} \;\text{such that}\;
\|x - x_0\|_p \le \varepsilon
\;\wedge\;
\mathcal{P}(x_0, x).
\]

\begin{definition}[Local Robustness Radius]
Let $f : \mathbb{R}^n \rightarrow \mathbb{R}^m$ be a neural network and
$x_0 \in \mathbb{R}^n$ a reference input.
Let $\mathcal{P}(x_0, x)$ denote a task-specific predicate expressing output
inconsistency.
The \emph{local robustness radius} at $x_0$ is defined as
\[
\varepsilon^\star \;=\;
\inf \left\{
\varepsilon \ge 0 \;\middle|\;
\exists x \in \mathbb{R}^n,\;
\|x - x_0\| \le \varepsilon
\;\wedge\;
\mathcal{P}(x_0, x)
\right\}.
\]
\end{definition}

The task is to compute certified bounds
$\underline{\varepsilon} \le \varepsilon^\star \le \overline{\varepsilon}$
such that
$\overline{\varepsilon} - \underline{\varepsilon} \le \delta$ for a given
precision $\delta > 0$.
The lower bound corresponds to a formally verified robustness guarantee
(\textsf{UNSAT}), while the upper bound corresponds to a radius at which a
violating input is found (\textsf{SAT}).

Local robustness radius computation is a standard benchmark in neural network
verification and provides a quantitative measure of model stability around a
given input.

A standard approach to this task repeatedly invokes a neural network verifier
over a sequence of verification queries, typically using a binary-search style
procedure over candidate perturbation radii.
Each query examines a radius $\varepsilon_i$ within a prescribed interval:
an \textsf{UNSAT} result certifies robustness at $\varepsilon_i$ and yields a new
lower bound, while a \textsf{SAT} result produces a violating input and yields a
new upper bound.
The next candidate radius is selected accordingly, and the process continues
until the remaining interval is within the desired precision or the
computational budget is exhausted.
A precise procedural description is given in
Appendix~\ref{app:robustness-radius}.

This procedure induces a sequence of verification queries
$q_1, q_2, \ldots, q_k$, all defined over the same network $f$ and reference input
$x_0$.
Each query $q_i$ corresponds to a perturbation radius $\varepsilon_i \ge 0$ and
checks the output predicate $\mathcal{P}(x_0,x)$ for all inputs satisfying
$\|x - x_0\| \le \varepsilon_i$.
Across the sequence, the perturbation radius is the only varying component and
determines the input constraint of the $i$-th query.

\begin{proposition}[Robustness Query Refinement]
\label{prop:robustness-query-refinement}
Let $q_i$ and $q_j$ be two robustness verification queries with perturbation
radii $\varepsilon_i > \varepsilon_j$, respectively.
Then $q_j$ is a refinement of $q_i$, denoted $q_j \preceq q_i$.
\end{proposition}

\noindent
Proposition~\ref{prop:robustness-query-refinement} shows that robustness radius
determination induces a refinement-ordered family of verification queries.
The proof is provided in Appendix~\ref{app:usecase-refinement}.

In the incremental robustness radius determination procedure, each verification
query $q_i$ inherits conflicts learned from all previously issued queries
$q_j$ with larger perturbation radii $\varepsilon_j > \varepsilon_i$.
This allows later queries, which impose stricter input constraints, to reuse
conflicts learned under looser constraints and prune infeasible regions early.
Since the binary search explores radii in a non-monotonic order, conflicts are 
selectively inherited only from queries with $\varepsilon_j > \varepsilon_i$, which 
are guaranteed to be valid refinements by Proposition~\ref{prop:robustness-query-refinement}. 

\subsubsection{Evaluation}

We evaluated the effectiveness of incremental conflict reuse for determining
the local robustness radius.
We compared our incremental approach against a non-incremental baseline on the
MNIST dataset, using a fully connected neural network from the VNN-COMP
benchmark~\cite{vnncomp1,vnncomp21}\footnote{\url{https://github.com/VNN-COMP/vnncomp2021/blob/main/benchmarks/mnistfc/mnist-net_256x2.onnx}}.
For each experiment, we computed the local robustness radius for
$1000$ inputs from the MNIST test set using a precision parameter
$\delta = 0.001$.

Table~\ref{tab:robustness-radius-results} summarizes the results.
The \emph{Time} column reports the average total solving time 
per task (in seconds), \emph{Solved} denotes the number of inputs for which the robustness
radius was determined within the timeout, \emph{Propagations} reports the
average number of propagations induced by inherited conflicts per task, and
\emph{Conflicts} reports the average number of conflicts recorded per
task.
Overall, incremental verification with conflict reuse significantly outperforms
the baseline.
On average, the incremental method achieves a $1.3\times$ speedup in per-task
solving time, substantially reducing the total time required to determine
robustness radii across test inputs.
\begin{table}
\centering
\begin{tabular}{lcccc}
\toprule
\;\textbf{Method}\; &
\;\textbf{Time (s)}\;  &
\;\textbf{Solved}\;  &
\;\textbf{Propagations}\;  &
\;\textbf{Conflicts} \; \\
\midrule
Non-incremental &
315.6 &
160 &
-- &
-- \\
Incremental &
\bf{233.5} &
\bf{185} &
8.2 &
107.4\\
\midrule
\textbf{Speedup} &
\textbf{1.35$\times$} &
-- &
-- &
-- \\
\bottomrule
\end{tabular}
\caption{Robustness radius evaluation on MNIST.}
\label{tab:robustness-radius-results}
\end{table}

\begin{wrapfigure}{r}{0.48\textwidth}
\centering
\vspace{-0.5em}
\includegraphics[width=0.95\linewidth]{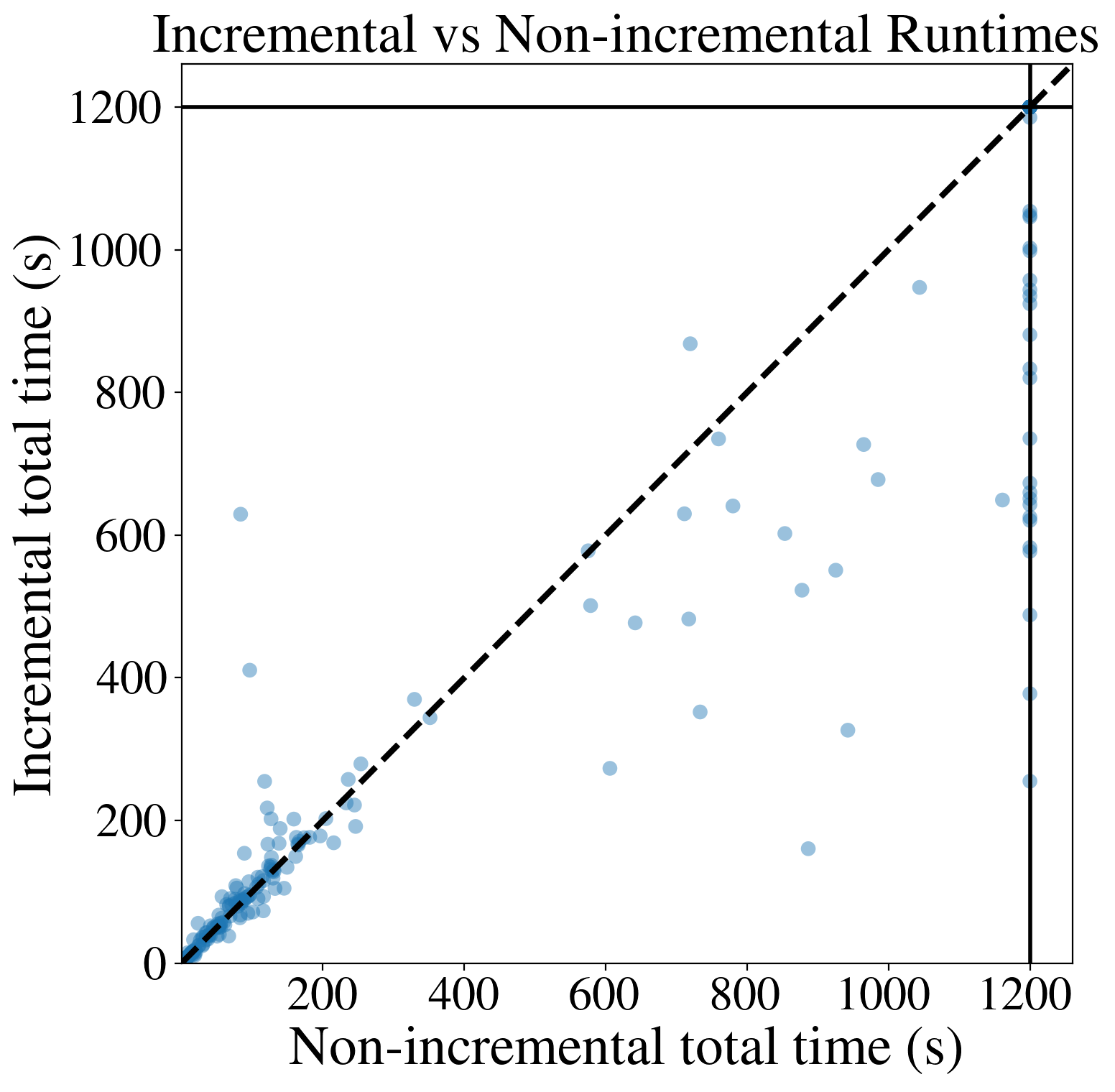}
\caption{Incremental vs.\ non-incremental for robustness radius determination.}
\label{fig:robustness-results}
\vspace{-5em}
\end{wrapfigure}
Figure~\ref{fig:robustness-results} presents a visual comparison of runtimes for
incremental and non-incremental robustness radius determination.
Each point corresponds to a single input instance.
Most points lie below the $y=x$ diagonal, with several non-incremental timeouts,
indicating better performance of the incremental approach on harder instances.

\subsection{Use Case 2: Verification with Input Splitting}
\label{sec:input-splitting}
Input splitting is a common approach for verifying challenging properties, often
referred to as split-and-conquer or divide-and-conquer verification.
Under input splitting, the input region is partitioned into multiple subregions,
and each subregion is verified independently.
As shown by Wu et al.~\cite{input_splitting}, if the verifier is sound and complete
and the subregions collectively cover the original input region, the overall
verification result can be determined compositionally.
That is, the property is \textsf{SAT} if at least one subregion is proven
\textsf{SAT}, and \textsf{UNSAT} if all subregions are proven \textsf{UNSAT}.

Input splitting is particularly important for scaling verification to difficult
properties where the full input region is too large to verify at once.
By recursively partitioning challenging regions into smaller subregions,
verifiers can verify properties that would otherwise time out.

With input splitting, each child query produced by splitting is a refinement of
the parent query that was split.
This follows because the input region of each child query is strictly contained
within the input region of the parent query.
The following proposition formalizes this observation.

\begin{proposition}[Input Splits are Refinements]
\label{prop:input-splitting-refinement}
Let $q_i$ be a verification query, and suppose
$q_{i+1}, \ldots, q_{i+n}$ are generated by input splitting of $q_i$.
Then each $q_{i+1}, \ldots, q_{i+n}$ is a refinement of $q_i$.
\end{proposition}

\noindent
Proposition~\ref{prop:input-splitting-refinement} shows that input splitting
induces refinement chains, enabling incremental verification across recursive
partitioning.
The proof is provided in Appendix~\ref{app:usecase-refinement}.

The input-splitting incremental verification procedure leverages
Proposition~\ref{prop:input-splitting-refinement} to enable incremental verification across
a recursive partitioning of the input space.
Given a verification query $q$ with input region $\mathcal{X}$, the verifier
first attempts to verify $q$ directly.
If verification terminates with a conclusive result (\textsf{SAT} or
\textsf{UNSAT}), that result is returned.
If verification times out, the input region $\mathcal{X}$ is partitioned by
splitting the largest interval dimension at its midpoint, producing two child
queries $q_{\text{left}}$ and $q_{\text{right}}$.

By Proposition~\ref{prop:input-splitting-refinement}, both child queries are
refinements of the parent query $q$.
Therefore, all conflict clauses learned during the attempted verification of
$q$ are retained and reused when verifying $q_{\text{left}}$ and
$q_{\text{right}}$.
Each child query is then verified recursively following the same procedure: if a
child query times out, it is further split and inherits all conflict clauses
from its ancestors.
As input splitting progresses, the set of learned conflict clauses therefore
monotonically increases.

The procedure continues until all leaf subregions are resolved or a specified
time limit is reached.
If all leaf subregions verify as \textsf{UNSAT}, the original property is
\textsf{UNSAT}.
Dually, if any leaf subregion verifies as \textsf{SAT}, the original property is
\textsf{SAT}.
The pseudocode for the input-splitting workflow is provided in
Appendix~\ref{app:input-split}.

\subsubsection{Evaluation}

To evaluate the effectiveness of incremental conflict reuse for input splitting,
we applied our approach to a counterexample-guided inductive synthesis (CEGIS)
loop for training Lyapunov neural certificates for deep reinforcement
learning-controlled spacecraft, as introduced by Mandal et al.~\cite{lyapunov}.
This framework iteratively learns a Lyapunov function that certifies
reach-while-avoid properties for a 4D spacecraft docking system.
The CEGIS loop alternates between training the certificate to satisfy Lyapunov
conditions and formally verifying these conditions using neural network
verification.
When verification identifies counterexamples that violate the Lyapunov
conditions, the counterexamples are added to the training data and the
certificate is retrained.

We executed the complete CEGIS training procedure and extracted the $680$
verification queries generated during the process.
Of these, $189$ queries were solved without requiring input splitting and were
excluded from the evaluation, as conflict reuse only affects verification once
input splitting is performed.
The remaining $491$ queries form the basis of our evaluation.
For these queries, we employed a progressive timeout strategy: the initial
verification attempt was run with a $5$-second timeout, and each subsequent input
split increased the timeout by a factor of $1.5$ (i.e., $7.5$ seconds after the
first split, $11.25$ seconds after the second, and so on).
A global timeout of $1200$ seconds was imposed to prevent indefinite splitting.

Table~\ref{tab:input-split-results} summarizes the performance of the incremental
and non-incremental methods.
The \emph{Time} column reports the average verification time per
verification task in
seconds.
\emph{Solved} denotes the number of verification tasks completed within the
global timeout.
\emph{Propagations} reports the average number of propagations induced per
verification task, and \emph{Conflicts} reports the average number of conflict
clauses recorded during verification.

\begin{table}[h]
\centering
\begin{tabular}{lcccc}
\toprule
\;\textbf{Method}\;  &
\;\textbf{Time (s)}\;  &
\;\textbf{Solved}\;  &
\;\textbf{Propagations}\;  &
\;\textbf{Conflicts} \; \\
\midrule
Non-incremental &
84.1 &
489 &
-- &
-- \\
Incremental &
\textbf{43.9} &
\textbf{491} &
1.7 &
7.9\\
\midrule
\textbf{Speedup} &
\textbf{1.92$\times$} &
-- &
-- &
-- \\
\bottomrule
\end{tabular}
\vspace{0.5em}
\caption{Evaluation of iterative input splitting for Lyapunov certificate verification.}
\label{tab:input-split-results}
\end{table}

Overall, the incremental method significantly outperforms the non-incremental
baseline.
It successfully solves all $491$ verification tasks, whereas the baseline times
out on two tasks, and achieves an average speedup of $1.92\times$.
On average, each verification attempt involved $7.9$ recorded conflict clauses,
which induced $1.7$ propagations per attempt.

\begin{wrapfigure}{r}{0.45\textwidth}
\vspace{-2.3em}
\centering
\includegraphics[width=0.95\linewidth]{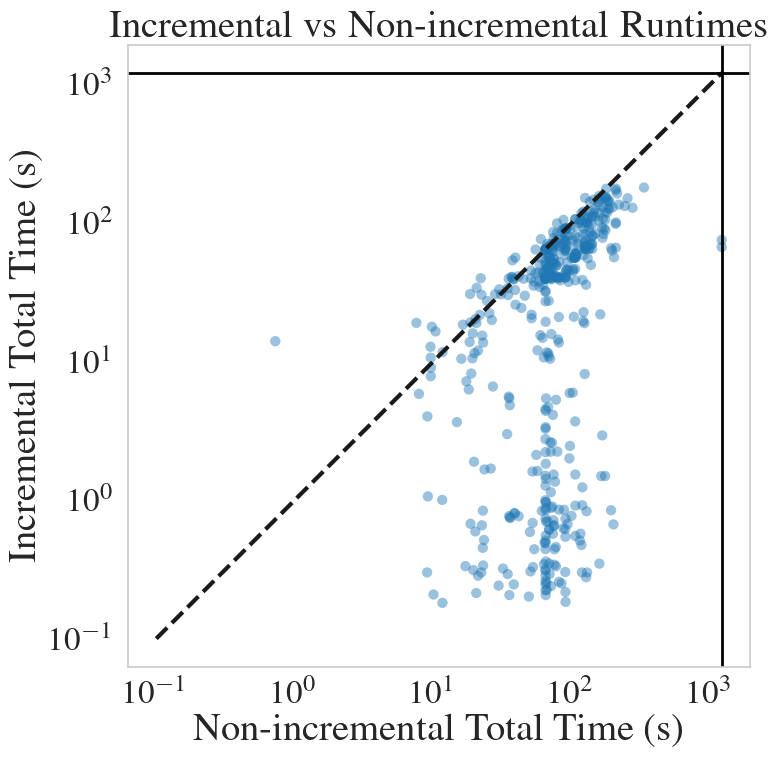}
\caption{
Incremental vs.\ non-incremental for input splitting.
%The dashed line indicates equal performance, and the solid lines represent the
%timeout threshold (1200 seconds).
}
\label{fig:inputsplit-plot}
\vspace{-4.5em}
\end{wrapfigure}
Figure~\ref{fig:inputsplit-plot} presents a visual comparison of runtimes for both
methods on a logarithmic scale.
Each point corresponds to a single verification query.
Most points lie below the 
equal-performance line, indicating that incremental
verification outperforms the non-incremental baseline on the majority of queries.

\subsection{Use Case 3: Minimal Sufficient Feature Set Extraction}
\label{sec:msfs-extraction}
Minimal sufficient feature set extraction addresses the problem of explaining a
neural network's prediction by identifying a smallest subset of input features
whose values alone suffice to determine the
output~\cite{VeriX,verix_plus,towards_xai,joao,BoElDuBaKa26}.

Given a reference input and its predicted class, the goal is to identify a subset
of input indices such that fixing these features to their reference values
guarantees preservation of the predicted class, regardless of how the remaining
features vary.
While this notion naturally extends to regression settings via a tolerance on the
output, we focus here on classification for clarity.

Formally, we seek a feature subset that is sufficient for preserving the
prediction and minimal with respect to set inclusion, subject to a given time
budget.
In this work, unfixed features are allowed to vary freely over their entire
domain.

\begin{definition}[Minimal Sufficient Feature Set]
Let $f : \mathbb{R}^n \rightarrow \mathbb{R}^m$ be a neural network,
$x_0 \in \mathbb{R}^n$ a reference input, and
$c = \arg\max_{j} f_j(x_0)$ the predicted class at $x_0$.
For a subset of feature indices $S \subseteq \{1,\dots,n\}$, define
\[
\mathcal{X}_S \;=\;
\left\{
x \in \mathbb{R}^n \;\middle|\;
x_i = (x_0)_i \;\text{for all}\; i \in S
\right\}.
\]
The set $S$ is a \emph{sufficient feature set} if
\[
\forall x \in \mathcal{X}_S,\quad
\arg\max_j f_j(x) = c.
\]
It is \emph{minimal} if no strict subset $S' \subset S$ is sufficient.
\end{definition}

Minimal sufficient feature sets reveal which input features are essential for a
given prediction and form a central primitive in formal explainability.

A common approach to solving this task is to search over subsets of input
features while invoking a neural network verifier as a subroutine.
The procedure begins with all features fixed to their values in the reference
input
\begin{wrapfigure}{r}{0.42\textwidth}
 \vspace{0.5ex}
 \centering
 \includegraphics[width=\linewidth]{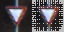}
 \vspace{0.5ex}
 \includegraphics[width=\linewidth]{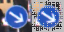}
\caption{Minimal sufficient feature set visualizations for GTSRB inputs.
Pixels not included in the explanation are shown in gray.}
 \label{fig:gtsrb-expl-visual}
 \vspace{-5.5ex}
\end{wrapfigure}
 and progressively frees features to test whether preservation of the
predicted class is maintained.
At a high level, the procedure alternates between proposing candidate feature
sets to free and using verification outcomes to determine whether these features
can be freed or must remain fixed.
This process follows an \emph{anytime} paradigm: if interrupted due to a timeout,
it returns the smallest sufficient feature set identified so far.

To explore the space of feature subsets efficiently, an importance ordering over
features is typically computed in a preprocessing step and used to guide a
binary-search-style exploration.
Rather than freeing features individually, the procedure considers groups of
features---often starting from the least important ones---and recursively refines
them.
Candidate sets are tentatively freed and checked for sufficiency; depending on
the outcome, features are either permanently freed or selectively reinstated.

A precise procedural description of this workflow, including the technical
modifications required to support incremental verification and the binary-search
strategies employed, is provided in Appendix~\ref{app:explainability}.

Under this formulation, the task gives rise to a collection of verification
queries, all defined over the same neural network $f$ and reference input $x_0$.
Each query $q$ is parameterized by a candidate set of freed features
$\overline{S}$, inducing the fixed feature set
$S = \{1,\dots,n\} \setminus \overline{S}$ and the corresponding constrained input
set $\mathcal{X}_S$.
The query asks whether fixing all features in $S$ to their values in $x_0$ is
sufficient to preserve the predicted class $c$.
A query returning \textsf{UNSAT} certifies that $S$ is sufficient and that all
features in $\overline{S}$ can be freed.
In contrast, a \textsf{SAT} result indicates that $S$ is insufficient and that
additional features must be fixed.
For soundness, \textsf{TIMEOUT} outcomes are treated as \textsf{SAT}.

All queries share the same network $f$, reference input $x_0$, and target class
$c$; the only varying component is the set of fixed features.
As a result, minimal sufficient feature set extraction induces a family of
closely related verification queries whose input constraints differ only in which
features are fixed.

Due to the binary-search-style exploration, the verification queries generated
during minimal sufficient feature set extraction are naturally organized as a
search tree rather than a linear sequence.
Queries generated after a \textsf{SAT} or \textsf{TIMEOUT} outcome impose strictly
stronger input constraints by fixing additional features.
In contrast, queries generated after a \textsf{UNSAT} outcome test disjoint sets
of features and are therefore incomparable under refinement.
Consequently, the overall query family is not totally ordered by refinement.

We formalize the refinement relationship that arises along \textsf{SAT} and
\textsf{TIMEOUT} branches of the search tree below.

\begin{proposition}[Feature-Set Query Refinement]
\label{prop:feature-set-query-refinement}
Let $q$ and $q'$ be two verification queries generated during the
feature set extraction procedure, corresponding to fixed feature sets
$S$ and $S'$, respectively.
If $q'$ is generated from $q$ following a \textsf{SAT} or \textsf{TIMEOUT}
outcome, then $S \subset S'$ and $q'$ is a refinement of $q$, denoted
$q' \preceq q$.
\end{proposition}

\noindent
Proposition~\ref{prop:feature-set-query-refinement} shows that, although the full
set of queries is not totally ordered, all queries along \textsf{SAT} and
\textsf{TIMEOUT} branches form refinement chains.
The proof is provided in Appendix~\ref{app:usecase-refinement}.

Accordingly, in the incremental minimal sufficient feature set extraction
procedure, each query $q'$ generated after a \textsf{SAT} or \textsf{TIMEOUT}
outcome inherits all conflicts learned from its ancestor queries along the same
search-tree branch.

\subsubsection{Evaluation}
We evaluated the effectiveness of incremental conflict reuse for accelerating
minimal sufficient feature set extraction on the German Traffic Sign Recognition
Benchmark (GTSRB) dataset~\cite{gtsrb}, using a convolutional neural network model
from Wu et al.~\cite{verix_plus}.
Out of the $1000$ test inputs considered, $70$ triggered conflicts during
\textsf{SAT} or \textsf{TIMEOUT} verification outcomes; the remaining inputs were
resolved without conflicts and therefore could not benefit from conflict reuse.
We restrict our evaluation to these $70$ cases.

For each input, we ran an \emph{anytime} minimal sufficient feature set extraction
procedure and compared our incremental approach against a non-incremental
baseline.
Table~\ref{tab:feature-set-extraction-results} summarizes the results.
The \emph{Explanation Size} column reports the average size of the minimal
sufficient feature set returned within the time budget.
\emph{Propagations} reports the average number of propagations induced by
inherited conflicts per task, and \emph{Conflicts} reports the average number of
conflict clauses recorded per task.
Both approaches achieve comparable explanation sizes, with the incremental method
yielding slightly smaller explanations on average.
On average, the incremental approach records $92$ conflicts per task, which
induce approximately $2$ effective propagations per task.

\begin{table}[h]
\centering
\begin{tabular}{lcccc}
\toprule
\;\textbf{Method}\; &
\;\textbf{Explanation Size}\; &
\;\textbf{Propagations}\; &
\;\textbf{Conflicts} \;\\
\midrule
Non-incremental &
848.52 &
-- &
-- \\
Incremental &
\bf{844.21} &
2.30 &
92.14 \\
\bottomrule
\end{tabular}
\caption{Minimal sufficient feature set extraction on GTSRB.}
\label{tab:feature-set-extraction-results}
\end{table}
\vspace{-1.5em}

The primary benefit of incremental verification in this setting lies in its
\emph{anytime} behavior.
As shown in Figure~\ref{fig:feature-set-extraction-results}, incremental
verification progressively outperforms the non-incremental baseline in reducing
explanation size over time.
During the initial phase (up to approximately 20 seconds), both methods exhibit
similar performance, with the non-incremental approach sometimes slightly ahead,
reflecting the overhead incurred while recording conflicts early in the search.

\begin{wrapfigure}{r}{0.55\textwidth}
\vspace{-0.5em}
\centering
\includegraphics[width=0.95\linewidth]{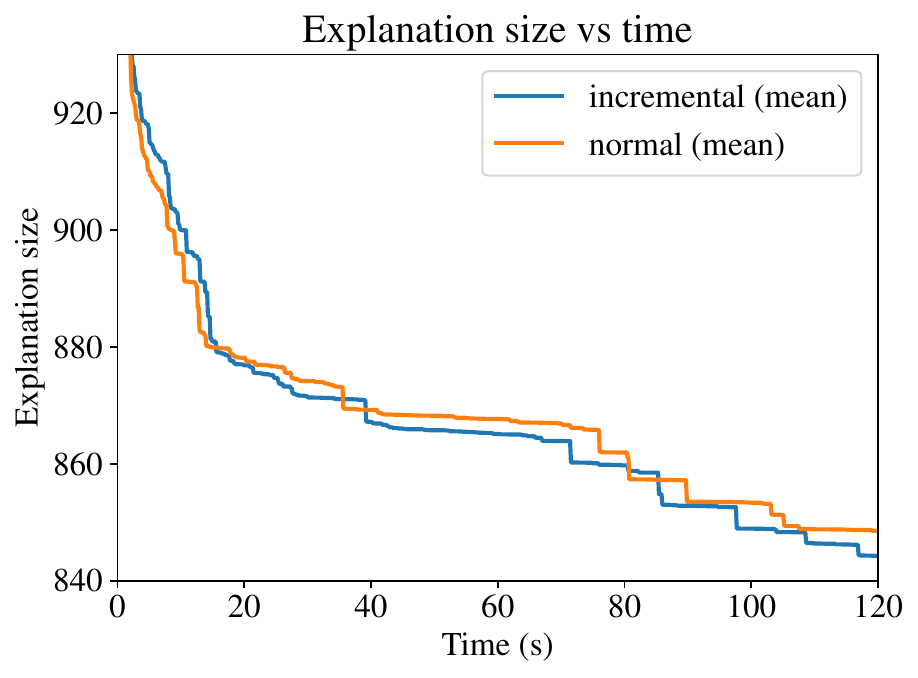}
\caption{
Incremental vs.\ non-incremental explanation size over time.
}
\label{fig:feature-set-extraction-results}
\vspace{-2em}
\end{wrapfigure}

Beyond this point, incremental verification tends to achieve smaller
explanations more quickly by reusing previously learned conflicts accumulated
earlier in the search.
This reuse supports earlier identification of critical features and more
efficient pruning of infeasible feature combinations, leading to improved
anytime behavior.
Overall, these observations suggest that conflict reuse can be beneficial for
explainability tasks under an anytime verification regime.

\subsection{Results Discussion}

We evaluated incremental conflict reuse on three verification tasks: robustness
radius determination, input splitting, and minimal sufficient feature set
extraction.

In robustness radius determination, queries form a refinement chain,
allowing conflicts learned at larger radii to be reused in subsequent queries.
This yields a 26\% reduction in average verification time compared to the
non-incremental baseline (from 315.6 to 233.5 seconds), corresponding to a
$1.35\times$ speedup.
In input splitting, queries form refinement chains along recursive split
branches, allowing conflicts learned in parent subproblems to be reused
consistently in descendant queries.
This yields a 47\% reduction in average verification time compared to the
non-incremental baseline (from 84.1 to 43.9 seconds), corresponding to a
$1.92\times$ speedup.
For minimal sufficient feature set extraction, conflict reuse occurs along
refinement chains induced by \textsf{SAT} and \textsf{TIMEOUT} outcomes.
While final explanation size improvement is small, incremental reuse improves anytime behavior
by reducing explanation size more quickly.

Overall, the impact of incremental conflict reuse is closely tied to the
refinement structure of the query family: stronger refinement relations enable
greater reuse and larger performance gains.

\section{Related Work}
\label{sec:related-work}

Incremental SAT and SMT solving address scalability by reusing learned
information, such as conflict clauses and theory lemmas, across sequences of
related problem instances~\cite{incremental_sat_1,smt_handbook,inc_smt}.
When successive instances are structurally similar, this can avoid redundant
reasoning and improve performance.
But, the effectiveness of incremental techniques is strongly
problem-dependent, and worst-case results in areas such as dynamic graph
algorithms show that substantial recomputation may be unavoidable under general
updates~\cite{HoDeTh01}.
Despite the maturity of incremental SAT/SMT solving, its systematic application
to neural network verification remains limited.

Prior work on incremental neural network verification has largely focused on settings in
which the network itself changes across queries.
Residual reasoning~\cite{residual_reasoning} reuses information learned for
abstract networks to accelerate verification after refinement, while IVAN and
I-IVAN~\cite{ivan_2,ivan} reuse successful case splits heuristically across
related network architectures.
More recently, Zhang et al.~\cite{inc_counterexample} study incremental
verification guided by counterexample potentiality, again in scenarios where the
network is modified.

Related ideas have been explored in abstract interpretation.
FANC~\cite{proof_transfer} employs heuristic transfer of abstract bounds to certify
multiple approximate neural networks, whereas we reuse conflict information
across verification queries on a fixed network.

In contrast, we consider iterative verification over a fixed neural network,
where properties or input constraints vary across queries.
Our approach reuses learned conflict clauses and provides conditions under which
they can be transferred soundly between related verification queries, aligning
with conflict-based incremental SAT/SMT solving applied to a single
network.

In our evaluation, we instantiated our approach with the CaDiCaL
SAT solver~\cite{cadical} and the Marabou verifier~\cite{marabou,marabou2}
as backends. CaDiCaL is a modern SAT
solver with a clean design that facilitates its integration with
other tools. Marabou is a proof-producing DNN analysis
framework~\cite{IsBaZhKa22}, which has been used for a myriad of
applications, including network pruning~\cite{LaKa21}, formal
explainability~\cite{BaAmCoReKa23}, verifying network
generalization~\cite{AmMaZeKaSc23}, and network
ensembles~\cite{AmZeKaSc22}. 
While our implementation relies on Marabou and CaDiCaL, the proposed
technique is solver-agnostic and can, in principle, be integrated with
other SAT solvers and branch-and-bound verifiers.

\section{Limitations and Future Work}
\label{sec:future-work}

The current implementation records conflicts without requiring minimality.
As a result, a conflict may include ReLU phase decisions that are not strictly
necessary to establish infeasibility.
While smaller or subsumed conflicts could improve reuse effectiveness and reduce
the overhead of SAT-based reasoning, computing minimal conflicts would require
additional analysis and is left for future work~\cite{minimization2,crown_min2_BICCOS,crown_min}.

More generally, the approach focuses on reusing conflicts derived from infeasible
subproblems.
Other reusable information, such as richer theory-specific lemmas or
abstractions, may further improve performance in some settings.
Exploring such mechanisms would require careful consideration of both soundness
and solver integration.

Finally, learned conflicts are used solely for pruning and propagation.
An additional direction is to exploit conflicts to guide branching decisions,
for example by prioritizing frequently occurring case splits.

\section{Conclusion}
\label{sec:conclusion}

We introduced an incremental verification approach for neural networks that
reuses learned conflict clauses across related verification queries over a fixed
network.
By formalizing a refinement relation between queries, we showed when such
conflicts can be reused soundly and integrated this mechanism into the Marabou
verifier as a lightweight extension to branch-and-bound search.
Our evaluation shows that conflict reuse reduces redundant exploration and yields
speedups of up to $1.9\times$ on representative iterative verification tasks,
while preserving soundness.

\section*{Acknowledgments}

Raya Elsaleh is supported by the Ariane de Rothschild Women Doctoral Program.
The work of Elsaleh and Katz was partially funded by a grant from the
Israeli Science Foundation (grant number 558/24), and by the European
Union (ERC, VeriDeL, 101112713). Views and opinions expressed are
however those of the author(s) only and do not necessarily reflect
those of the European Union or the European Research Council Executive
Agency. Neither the European Union nor the granting authority can be
held responsible for them.

\bibliographystyle{abbrv}
\bibliography{references}

\newpage
\appendix
\section*{Appendix}
\addcontentsline{toc}{section}{Appendix}

\section{Additional Implementation Details}
\label{app:implementation-details}

\subsubsection{Push-Pop of Input and Output Constraints.}
To improve efficiency, we introduce a push-pop mechanism for managing input and
output constraints in Marabou.
This mechanism enables reuse of the encoded neural network and constraints across 
verification queries by allowing constraints to be incrementally added and
removed without re-encoding the network, the constraints, or re-initializing the solver.
Using \emph{push}, the verifier can lock a set of constraints and incrementally
add additional ones; a subsequent \emph{pop} restores the solver state to the
previous configuration.

\section{Soundness of Conflict Reuse}
\label{app:conflict-soundness}

This appendix contains the proofs of Lemma~\ref{lem:monotonicity} and
Theorem~\ref{thm:conflict-reuse} from
Section~\ref{sec:incremental-verification}.

\subsection{Proof of Lemma~\ref{lem:monotonicity}}
\begin{proof}
Assume for contradiction that
$q_2 \land \ell_1 \land \cdots \land \ell_k$
is feasible, and let $x$ be a satisfying input.
Since $q_2 \preceq q_1$, both queries are defined over the same network $f$,
$\mathcal{X}(q_2) \subseteq \mathcal{X}(q_1)$, and the ReLU phase literals
$\ell_1,\ldots,\ell_k$ refer to the same ReLU constraints in both queries.
Therefore, the same input $x$ satisfies
$q_1 \land \ell_1 \land \cdots \land \ell_k$,
contradicting infeasibility.
\qed
\end{proof}

\subsection{Proof of Theorem~\ref{thm:conflict-reuse}}
\begin{proof}
Let $c=\{\ell_1,\ldots,\ell_k\}$ be a conflict clause for $q_1$.
By Definition~\ref{def:conflict-clause}, the subproblem
$q_1 \land \ell_1 \land \cdots \land \ell_k$
is infeasible.
Since $q_2 \preceq q_1$, infeasibility is monotone under refinement by
Lemma~\ref{lem:monotonicity}.
Therefore, the subproblem
$q_2 \land \ell_1 \land \cdots \land \ell_k$
is also infeasible, and hence $c$ is a conflict clause for $q_2$.
\qed
\end{proof}

\section{Soundness of SAT-Based Conflict Reasoning}
\label{app:sat-soundness}

This appendix provides proofs for Lemmas~\ref{lem:sat-pruning}
and~\ref{lem:sat-implied}, which establish the soundness of SAT-based pruning and
propagation using learned conflict clauses.

\subsection{Proof of Lemma~\ref{lem:sat-pruning}}
\begin{proof}
Suppose, for contradiction, that there exists a feasible solution for query $q$
that extends the partial assignment $\alpha$.
Then there exists a complete assignment that satisfies all literals in $\alpha$
and corresponds to a feasible input for $q$.

Since the propositional CNF formula constructed from $\alpha$ and the clauses in
$\mathcal{C}$ is \textsf{UNSAT}, every complete assignment extending $\alpha$ must violate
at least one conflict clause $c \in \mathcal{C}$.
By Definition~\ref{def:conflict-clause}, each conflict clause forbids the
simultaneous satisfaction of all its literals, and any assignment violating $c$
cannot correspond to a feasible solution for $q$.
This contradicts the assumption of feasibility.
Therefore, no extension of $\alpha$ can correspond to a feasible solution for
$q$.
\qed
\end{proof}

\subsection{Proof of Lemma~\ref{lem:sat-implied}}

\begin{proof}
Let $\mathcal{L}$ be the set of literals implied by unit propagation when solving
the SAT instance under assumptions $\alpha$.
By definition of unit propagation, for each $\ell \in \mathcal{L}$, assigning
$\neg \ell$ under $\alpha$ would falsify at least one conflict clause
$c \in \mathcal{C}$.

Suppose, for contradiction, that there exists a feasible solution for $q$ that
extends $\alpha$ but violates some implied literal $\ell \in \mathcal{L}$.
Then the corresponding complete assignment satisfies $\alpha \cup \{\neg \ell\}$
and therefore violates the conflict clause $c$.
By Definition~\ref{def:conflict-clause}, this assignment cannot correspond to a
feasible solution for $q$, contradicting feasibility.

Hence, every feasible extension of $\alpha$ must satisfy all literals in
$\mathcal{L}$.
\qed
\end{proof}

\section{Proofs of Refinement Properties of Verification Use Cases}
\label{app:usecase-refinement}
This appendix provides proofs of the refinement properties for the verification
use cases discussed in Section~\ref{sec:verification-tasks}.

\subsection{Proof of Proposition~\ref{prop:robustness-query-refinement}}
\begin{proof}
The two queries are defined over the same network $f$ and reference input $x_0$.
The input constraints of the two queries are given by
\[
\mathcal{X}(q_i) \;=\;
\left\{ x \in \mathbb{R}^n \;\middle|\;
\|x - x_0\| \le \varepsilon_i
\right\},
\quad
\mathcal{X}(q_j) \;=\;
\left\{ x \in \mathbb{R}^n \;\middle|\;
\|x - x_0\| \le \varepsilon_j
\right\}.
\]
Since $\varepsilon_j < \varepsilon_i$, it follows immediately that
$\mathcal{X}(q_j) \subseteq \mathcal{X}(q_i)$.
The output constraints of both queries are identical, as they are defined by the
same predicate $\mathcal{P}(x_0, x)$.
Therefore, by Definition~\ref{def:refinement}, $q_j$ is a refinement of $q_i$.
\qed
\end{proof}

\subsection{Proof of Proposition~\ref{prop:input-splitting-refinement}}
\begin{proof}
Let $q_0$ be a verification query, and let $q_1$ be a verification query generated
by input splitting of $q_0$. By construction, input splitting partitions the
input domain of $q_0$ into subregions, and thus the input domain of $q_1$ is a
subset of that of $q_0$, i.e.,
\(
\mathcal{X}(q_1) \subseteq \mathcal{X}(q_0).
\)
Moreover, input splitting does not modify the network or the output constraints.
Hence, $\mathcal{Y}(q_1) = \mathcal{Y}(q_0)$, and in particular
$\mathcal{Y}(q_1) \subseteq \mathcal{Y}(q_0)$.
Therefore, by Definition~\ref{def:refinement}, $q_1$ is a refinement of $q_0$.
\qed
\end{proof}

\subsection{Proof of Proposition~\ref{prop:feature-set-query-refinement}}
\begin{proof}
By construction, a successor query $q'$ generated after a \textsf{SAT} or
\textsf{TIMEOUT} outcome is obtained by reinstating a subset of previously freed
features.
Thus, the corresponding fixed feature sets satisfy $S \subset S'$.

The input constraints of the two queries are given by
\begin{multline*}
\mathcal{X}_S =
\left\{
x \in \mathbb{R}^n \;\middle|\;
x_i = (x_0)_i \;\text{for all}\; i \in S
\right\}, \\
\mathcal{X}_{S'} =
\left\{
x \in \mathbb{R}^n \;\middle|\;
x_i = (x_0)_i \;\text{for all}\; i \in S'
\right\}.
\end{multline*}
Since $S \subset S'$, it follows that
$\mathcal{X}_{S'} \subseteq \mathcal{X}_S$.

Both queries enforce the same output constraint, namely preservation of the
predicted class $c$.
Therefore, by Definition~\ref{def:refinement}, $q'$ is a refinement of $q$.
\qed
\end{proof}

\section{Algorithmic Details and Pseudocodes}

\subsection{Robustness Radius Interval Computation}
\label{app:robustness-radius}
This appendix provides pseudocode and a description of the robustness radius
interval computation used in Section~\ref{sec:robustness-verification}.
We present the algorithmic workflow by which lower and upper bounds on the local
robustness radius are iteratively refined through repeated verification queries
until the desired precision is achieved or the time budget is exhausted.
In addition, we highlight the modifications required to enable incremental reuse
of learned conflicts across related verification queries.

\begin{algorithm}[t]
\caption{Robustness Radius Interval Computation with \inc{Incremental Reuse}}
\label{alg:robustness-radius-interval-inc}
\DontPrintSemicolon
\LinesNumbered

\KwIn{
Reference input $x_0$; initial interval $[\varepsilon_{\min}, \varepsilon_{\max}]$;
precision $\delta>0$; time budget $T$;
verification subroutine $\textsc{Verify}(\varepsilon)$ \inc{with optional incremental reuse of inherited conflicts}
}
\KwOut{
Certified bounds $(\underline{\varepsilon}, \overline{\varepsilon})$ such that
$\underline{\varepsilon} \le \varepsilon^\star \le \overline{\varepsilon}$ and
$\overline{\varepsilon}-\underline{\varepsilon}\le \delta$ if completed within $T$.
}

\BlankLine
\textbf{Interface:}
$\textsc{Verify}(\varepsilon, x_0)$ returns $(res, x_{\text{ce}})$ where
$res \in \{\textsf{UNSAT},\textsf{SAT},\textsf{TIMEOUT}\}$ and
$x_{\text{ce}} \neq \bot$ only if $res=\textsf{SAT}$.
\inc{In incremental mode, we call $\textsc{Verify}(\varepsilon, x_0, id, \mathcal{I})$ where $id$ is the query identifier and $\mathcal{I}$ is a set of inherited query identifiers.}\;
\BlankLine

$start \gets \textsc{Now}()$\;\label{ln:init-time}

$\underline{\varepsilon} \gets \varepsilon_{\min}$, \qquad
$\overline{\varepsilon} \gets \varepsilon_{\max}$\; \label{ln:init-bounds}
$stepRatio \gets 1/2$, \qquad
$direction \gets \textsf{down}$\;

\inc{$nextId \gets 0$ \tcp*{counter for verification-query identifiers}}

\BlankLine
\While{$(\overline{\varepsilon}-\underline{\varepsilon}) > \delta$ \textbf{and} $\textsc{Now}()-start < T$}{\label{ln:main-loop}
    $w \gets \overline{\varepsilon}-\underline{\varepsilon}$\;
    \uIf{$direction = \textsf{down}$}{
        $\varepsilon \gets \overline{\varepsilon} - stepRatio \cdot w$\;
    }\Else{
        $\varepsilon \gets \underline{\varepsilon} + stepRatio \cdot w$\;
    }

    \inc{$id \gets nextId;\;\; nextId \gets nextId+1$}\;
    \inc{$\varepsilon_{id} \gets \varepsilon$}\;

    \inc{$\mathcal{I} \gets \{\, id' \mid id' \text{ was issued earlier with } \varepsilon_{id'} > \varepsilon \,\}$}\;
    \label{ln:inheritance}
    $res,\; x_{\text{ce}} \gets \textsc{Verify}(\varepsilon, x_0\inc{,\, id,\, \mathcal{I}})$
    \label{ln:verify-call}
    \tcp*[r]{$x_{\text{ce}}$ is returned only if $res=\textsf{SAT}$}

    \uIf{$res=\textsf{UNSAT}$}{ \label{ln:unsat-case}
        $\underline{\varepsilon} \gets \varepsilon$\;
        $stepRatio \gets 1/2$; $direction \gets \textsf{down}$\;
    }
    \uElseIf{$res=\textsf{SAT}$}{ \label{ln:sat-case}
        $\overline{\varepsilon} \gets \|x_{\text{ce}} - x_0\|_\infty$
        \tcp*[r]{counterexample might yield a tighter bound}
        $stepRatio \gets 1/2$; $direction \gets \textsf{down}$\;
    }
    \Else(\tcp*[f]{\textsf{TIMEOUT}}){ \label{ln:timeout-case}
      \If{$direction = \textsf{down}$}{
          $stepRatio \gets stepRatio / 2$\;
          $direction \gets \textsf{up}$\;
      }
      \Else{
          $direction \gets \textsf{down}$\;
      }
    }
}

\BlankLine
\Return{$\underline{\varepsilon}, \overline{\varepsilon}$}\;\label{ln:return}
\end{algorithm}

Given a neural network $f : \mathbb{R}^n \rightarrow \mathbb{R}^m$, a reference input
$x_0 \in \mathbb{R}^n$, and an output consistency predicate
$\mathcal{P}(x_0, x)$, the objective is to compute certified lower and upper bounds
$\underline{\varepsilon}$ and $\overline{\varepsilon}$ on the local robustness radius
$\varepsilon^\star$ at $x_0$, up to a prescribed precision $\delta > 0$ and within a
time budget $T$.

The procedure is presented in Algorithm~\ref{alg:robustness-radius-interval-inc}.
It relies on a verification subroutine $\textsc{Verify}(...)$ that, given a
perturbation radius $\varepsilon$ and the reference input $x_0$, determines whether
the property $\mathcal{P}(x_0, x)$ holds for all inputs $x$ satisfying
$\|x - x_0\|_\infty \le \varepsilon$. In other words, it returns one of three possible outcomes:
\textsf{UNSAT}, \textsf{SAT}, or \textsf{TIMEOUT}.
In non-incremental mode, the call takes only the parameters
$\varepsilon$ and $x_0$ as input:
$\textsc{Verify}(\varepsilon, x_0)$.
In incremental mode, the call
$\textsc{Verify}(\varepsilon, x_0\inc{,\, id,\, \mathcal{I}})$ additionally takes a
unique query identifier $id$ and a set of inherited query identifiers
$\mathcal{I}$, enabling reuse of learned conflicts.
The elements related to incremental reuse are highlighted in blue in the
algorithm.

The algorithm assumes the initial bracketing interval
$[\varepsilon_{\min}, \varepsilon_{\max}]$ is valid (i.e.
$\textsc{Verify}(\varepsilon_{\min}, x_0)$ returns \textsf{UNSAT}, and
$\textsc{Verify}(\varepsilon_{\max}, x_0)$ returns \textsf{SAT}).
After initializing the time budget and bounds
(Lines~\ref{ln:init-time}--\ref{ln:init-bounds}),
the procedure iteratively refines the current interval
$[\underline{\varepsilon}, \overline{\varepsilon}]$ within a main loop
(Line~\ref{ln:main-loop}).

In each iteration, a candidate perturbation radius $\varepsilon$ is selected from
the current interval and issued as a verification query.
In incremental mode, the query is assigned a fresh identifier and augmented with a
set of inherited identifiers corresponding to earlier queries at larger radii
(Line~\ref{ln:inheritance}).
The verification subroutine is then invoked
(Line~\ref{ln:verify-call}).

The outcome of the verification determines how the interval is updated.
If the result is \textsf{UNSAT}, robustness is certified at radius $\varepsilon$ and
the lower bound $\underline{\varepsilon}$ is raised accordingly
(Line~\ref{ln:unsat-case}).
If the result is \textsf{SAT}, a counterexample is returned and used to tighten the
upper bound $\overline{\varepsilon}$ based on its distance from $x_0$
(Line~\ref{ln:sat-case}).
If the verification attempt times out, the algorithm conservatively adapts its
search direction and step size without updating the bounds
(Line~\ref{ln:timeout-case}).

The procedure terminates once the interval width
$\overline{\varepsilon} - \underline{\varepsilon}$ falls below the target precision
$\delta$, or when the available time budget is exhausted.
The final bounds $(\underline{\varepsilon}, \overline{\varepsilon})$ are then
returned (Line~\ref{ln:return}).

\subsection{Input-Split Verification with Incremental Reuse}
\label{app:input-split}
This appendix provides pseudocode and a description of the input-split verification
procedure used for neural network properties in Section~\ref{sec:input-splitting}.
We present the algorithmic workflow by which the input space is recursively
partitioned when verification queries time out, enabling the solver to make progress
on difficult instances.
In addition, we highlight the modifications required to enable incremental reuse
of learned conflicts across related verification queries.

\begin{algorithm}[t]
\caption{Input-Split Verification with \inc{Incremental Reuse}}
\label{alg:input-split-inc}
\DontPrintSemicolon
\LinesNumbered

\KwIn{
Neural network $f : \mathbb{R}^n \rightarrow \mathbb{R}^m$; property $\phi$; initial timeout $T_0$; timeout factor $\alpha$; verification subroutine $\textsc{Verify}(\mathcal{Q}, T)$ \inc{with optional incremental reuse of inherited conflicts}
}
\KwOut{
Verification result: $res \in \{\textsf{UNSAT},\textsf{SAT},\textsf{TIMEOUT}\}$.
}

\BlankLine
\textbf{Interface:}
$\textsc{Verify}(\mathcal{Q}, T)$ returns $res \in \{\textsf{UNSAT},\textsf{SAT},\textsf{TIMEOUT}\}$ where $\mathcal{Q}$ is a verification query with input bounds and $T$ is the timeout.
\inc{In incremental mode, we call $\textsc{Verify}(\mathcal{Q}, T, id, \mathcal{I})$ where $id$ is a query identifier and $\mathcal{I}$ is a set of inherited (ancestor) query identifiers.}\;
\BlankLine

\inc{$nextId \gets 0$ \tcp*{counter for query identifiers}}
$\mathcal{Q}_0 \gets \textsc{ConstructQuery}(f, \phi)$ \tcp*{initial query}
\Return{$\textsc{InputSplitSearch}(\mathcal{Q}_0, T_0\inc{,\, \emptyset})$}\;

\BlankLine
\Fn(){$\textsc{InputSplitSearch}(\mathcal{Q}, T\inc{,\, \mathcal{I}})$}{
    \inc{$id \gets nextId;\;\; nextId \gets nextId+1$}\;
    $res \gets \textsc{Verify}(\mathcal{Q}, T\inc{,\, id,\, \mathcal{I}})$\;\label{ln:verify-query}
    
    \uIf{$res=\textsf{SAT}$ or $res=\textsf{UNSAT}$}{\label{ln:conclusive}
        \Return{$res$}\;
    }
    
    $v \gets \textsc{SelectWidestInput}(\mathcal{Q})$\;\label{ln:select-input}
    
    $[\ell, u] \gets \textsc{GetBounds}(\mathcal{Q}, v)$\;
    $m \gets (\ell + u)/2$\;
    \inc{$\mathcal{I}' \gets \mathcal{I} \cup \{id\}$ \tcp*{inherit from parent}}
    
    $\mathcal{Q}_L \gets \textsc{TightenUpper}(\mathcal{Q}, v, m)$\;\label{ln:split-lower}
    $res_L \gets \textsc{InputSplitSearch}(\mathcal{Q}_L, \alpha \cdot T\inc{,\, \mathcal{I}'})$\;
    
    \If{$res_L = \textsf{SAT}$}{\label{ln:lower-sat}
        \Return{\textsf{SAT}}\;
    }
    
    $\mathcal{Q}_R \gets \textsc{TightenLower}(\mathcal{Q}, v, m)$\;\label{ln:split-upper}
    $res_R \gets \textsc{InputSplitSearch}(\mathcal{Q}_R, \alpha \cdot T\inc{,\, \mathcal{I}'})$\;
    
    \If{$res_R = \textsf{SAT}$}{\label{ln:upper-sat}
        \Return{\textsf{SAT}}\;
    }
    
    \uIf{$res_L = \textsf{UNSAT}$ and $res_R = \textsf{UNSAT}$}{\label{ln:both-unsat}
        \Return{\textsf{UNSAT}}\;
    }
    \Else{\label{ln:timeout}
        \Return{\textsf{TIMEOUT}}\;
    }
}
\end{algorithm}

Given a neural network $f : \mathbb{R}^n \rightarrow \mathbb{R}^m$, a property
$\phi$ to verify, the objective is to determine whether the property holds over the
input space. The procedure is presented in Algorithm~\ref{alg:input-split-inc}.
As in the previous use case, it relies on a verification subroutine
$\textsc{Verify}(...)$ that, given a verification query $\mathcal{Q}$ with input
bounds and a timeout $T$, determines whether the property holds for all inputs
in the constrained space.
In other words, it returns one of three possible outcomes:
\textsf{UNSAT}, \textsf{SAT}, or \textsf{TIMEOUT}.
In non-incremental mode, the call takes only the parameters
$\mathcal{Q}$ and $T$ as input:
$\textsc{Verify}(\mathcal{Q}, T)$.
In incremental mode, the call
$\textsc{Verify}(\mathcal{Q}, T\inc{,\, id,\, \mathcal{I}})$ additionally takes a
unique query identifier $id$ and a set of inherited query identifiers
$\mathcal{I}$, enabling reuse of learned conflicts.
The elements related to incremental reuse are highlighted in blue in the
algorithm.

The algorithm constructs an initial verification query $\mathcal{Q}_0$ from the
network $f$ and property $\phi$, then invokes the recursive search procedure
$\textsc{InputSplitSearch}$. At each node, the algorithm attempts to verify the current query $\mathcal{Q}$ with
timeout $T$ (Line~\ref{ln:verify-query}).
In incremental mode, the query is assigned a fresh identifier and augmented with a
set of inherited identifiers from ancestor queries.
If the verification returns a conclusive result (\textsf{SAT} or \textsf{UNSAT}),
that result is immediately returned (Line~\ref{ln:conclusive}).

When verification times out, the algorithm selects the input variable with the
widest valid range (Line~\ref{ln:select-input}) and splits its domain at the
midpoint, creating two disjoint subspaces.
The parent query's identifier is added to the inherited set $\mathcal{I}'$ for child
queries, enabling conflict reuse.

Two subproblems are created by tightening the upper bound (Line~\ref{ln:split-lower})
and lower bound (Line~\ref{ln:split-upper}) of the selected input variable.
Each recursive call uses increased timeout $\alpha \cdot T$.
If either branch finds a counterexample (\textsf{SAT}), the query returns
\textsf{SAT} (Lines~\ref{ln:lower-sat},~\ref{ln:upper-sat}).
If the \textsf{UNSAT} partitions collectively cover the entire input space, the
query is \textsf{UNSAT} (Line~\ref{ln:both-unsat}).
Otherwise, the result is \textsf{TIMEOUT} (Line~\ref{ln:timeout}).

\subsection{Minimal Sufficient Feature Set Extraction: Binary-Sequential Procedure}
\label{app:explainability}

This appendix provides pseudocode and a description of the recursive binary search
workflow used in Section~\ref{sec:msfs-extraction}.
The procedure repeatedly invokes a verification subroutine to determine whether a
candidate set of features can be freed while preserving the reference prediction.
The underlying binary search strategy follows~\cite{verix_plus} and is adapted here
to support incremental reuse of learned conflicts across related verification queries,
which arise along \textsf{SAT} and \textsf{TIMEOUT} branches of the search.

\begin{algorithm}[t]
\caption{Minimal Sufficient Feature Set Extraction with \inc{Incremental Reuse}}
\label{alg:binary-sequential-inc}
\DontPrintSemicolon
\LinesNumbered

\KwIn{
Neural network $f$; reference input $x_0 \in \mathbb{R}^n$; target class
$c = \arg\max_j f_j(x_0)$; feature universe $U=\{1,\dots,n\}$;
verification subroutine $\textsc{Verify}(\overline{S})$
\inc{with optional incremental reuse of inherited conflicts}.
}
\KwOut{
feature set that must remain fixed ($S_{\text{fixed}}$),
feature set certified freeable ($S_{\text{freed}}$).
}

\BlankLine
\textbf{Interface:}
$\textsc{Verify}(\overline{S}, x_0)$ returns $res \in \{\textsf{UNSAT},\textsf{SAT},\textsf{TIMEOUT}\}$ and
verifies if $U \setminus \overline{S}$ is a sufficient feature set for $x_0$.
\inc{In incremental mode, we call $\textsc{Verify}(\overline{S}, x_0, id, \mathcal{I})$ where $id$ is a query identifier and $\mathcal{I}$ is a set of inherited (ancestor) query identifiers.}\;
\BlankLine

$S_{\text{fixed}} \gets \emptyset$;\quad $S_{\text{freed}} \gets \emptyset$\;\label{ln:init-sets}
\inc{$nextId \gets 0$ \tcp*{counter for query identifiers}}

\BlankLine
$\textsc{BinaryMinimalSufficiencySearch}(U, \inc{\emptyset})$\;\label{ln:top-call}
\Return{$S_{\text{fixed}}, S_{\text{freed}}$}\;

\BlankLine
\Fn(){$\textsc{BinaryMinimalSufficiencySearch}(S_{\text{cand}}, \inc{\mathcal{I}})$}{
    \uIf{$|S_{\text{cand}}| = 1$}{\label{ln:base-case-1}
        \inc{$id \gets nextId;\;\; nextId \gets nextId+1$}\;
        $res \gets \textsc{Verify}(S_{\text{freed}} \cup S_{\text{cand}}, x_0\inc{,\, id,\, \mathcal{I}})$\;\label{ln:verify-single}
        \uIf{$res=\textsf{UNSAT}$}{\label{ln:single-unsat}
            $S_{\text{freed}} \gets S_{\text{freed}} \cup S_{\text{cand}}$\;
        }
        \Else(\tcp*[f]{\textsf{SAT} or \textsf{TIMEOUT}}){
            $S_{\text{fixed}} \gets S_{\text{fixed}} \cup S_{\text{cand}}$\;\label{ln:single-sat1}
        }
        \Return{}\;
    }

    $(S_{\text{left}}, S_{\text{right}}) \gets \textsc{Split}(S_{\text{cand}})$\;\label{ln:split}
    \inc{$id_{\text{left}} \gets nextId;\;\; nextId \gets nextId+1$}\;
    $res_{\text{left}} \gets \textsc{Verify}(S_{\text{freed}} \cup S_{\text{left}}, x_0\inc{,\, id_{\text{left}},\, \mathcal{I}})$\;\label{ln:verify-lo}

    \uIf{$res_{\text{left}}=\textsf{UNSAT}$}{\label{ln:lo-unsat}
        $S_{\text{freed}} \gets S_{\text{freed}} \cup S_{\text{left}}$\;

        \inc{$id_{\text{right}} \gets nextId;\;\; nextId \gets nextId+1$}\;
        $res_{\text{right}} \gets \textsc{Verify}(S_{\text{freed}} \cup S_{\text{right}}, x_0\inc{,\, id_{\text{right}},\, \mathcal{I}})$\;\label{ln:verify-hi}

        \uIf{$res_{\text{right}}=\textsf{UNSAT}$}{\label{ln:hi-unsat}
            $S_{\text{freed}} \gets S_{\text{freed}} \cup S_{\text{right}}$\;
        }
        \Else(\tcp*[f]{\textsf{SAT} or \textsf{TIMEOUT}}){\label{ln:hi-sat}
            $\textsc{BinaryMinimalSufficiencySearch}(S_{\text{right}}, \inc{\mathcal{I} \cup \{id_{\text{right}}\}})$\;
        }
    }
    \Else(\tcp*[f]{\textsf{SAT} or \textsf{TIMEOUT}}){\label{ln:lo-sat}
  
      \uIf{$|S_{\text{left}}| = 1$}{\label{ln:base-case-2}
          $S_{\text{fixed}} \gets S_{\text{fixed}} \cup S_{\text{cand}}$\;\label{ln:single-sat2}
      }
      \Else(){
          $\textsc{BinaryMinimalSufficiencySearch}(S_{\text{left}}, \inc{\mathcal{I} \cup \{id_{\text{left}}\}})$\;
      }
      $\textsc{BinaryMinimalSufficiencySearch}(S_{\text{right}})$\;
    }
}
\end{algorithm}

Given a neural network $f : \mathbb{R}^n \rightarrow \mathbb{R}^m$ and a reference
input $x_0 \in \mathbb{R}^n$, the goal is to extract a minimal sufficient feature
set for the predicted class $c=\arg\max_j f_j(x_0)$.
Equivalently, the goal is to partition the feature universe $U=\{1,\dots,n\}$ into
two disjoint sets: $S_{\text{fixed}}$, the set of features that must remain fixed
to their values in $x_0$, and $S_{\text{freed}}$, the set of features certified
freeable, such that fixing $S_{\text{fixed}}$ suffices to preserve the prediction
under arbitrary assignments to the features in $S_{\text{freed}}$.

The procedure is presented in Algorithm~\ref{alg:binary-sequential-inc}.
It relies on a verification subroutine $\textsc{Verify}(...)$ that, given a set
of candidate features $\overline{S}$ to free, determines whether
$U \setminus \overline{S}$ constitutes a sufficient feature set for $x_0$.
In other words, it returns one of three possible outcomes:
\textsf{UNSAT}, \textsf{SAT}, or \textsf{TIMEOUT}.
In non-incremental mode, the call takes only the parameters
$\overline{S}$ and $x_0$ as input:
$\textsc{Verify}(\overline{S}, x_0)$.
In incremental mode, the call
$\textsc{Verify}(\overline{S}, x_0\inc{,\, id,\, \mathcal{I}})$ additionally takes a
unique query identifier $id$ and a set of inherited query identifiers
$\mathcal{I}$, enabling reuse of learned conflicts.
The elements related to incremental reuse are highlighted in blue in the
algorithm.

Algorithm~\ref{alg:binary-sequential-inc} maintains the evolving sets
$S_{\text{fixed}}$ and $S_{\text{freed}}$ (Line~\ref{ln:init-sets}).
At each step it processes a candidate set of features $S_{\text{cand}}$ and attempts
to certify as many of them as freeable.
This is done by invoking the verification subroutine
$\textsc{Verify}(S_{\text{freed}} \cup S_{\text{cand}}, x_0)$
(Line~\ref{ln:verify-single}), which checks whether freeing the proposed features
preserves the predicted class $c$ while all remaining features are fixed to their
values in $x_0$.

An \textsf{UNSAT} result certifies that all features in the proposed set may be
freed, and the algorithm adds them to $S_{\text{freed}}$
(Lines~\ref{ln:single-unsat}, \ref{ln:lo-unsat}, and~\ref{ln:hi-unsat}).
A \textsf{SAT} or \textsf{TIMEOUT} result indicates that the proposed set cannot be
freed as a whole, and the algorithm refines the candidate by splitting it into
smaller subsets (Line~\ref{ln:split}) and continuing the search recursively on the
subset(s) that remain uncertified.

The recursion terminates at singletons (Lines~\ref{ln:base-case-1} and~\ref{ln:base-case-2}).
If a singleton feature cannot be freed, it is classified as necessary and added to
$S_{\text{fixed}}$ (Lines~\ref{ln:single-sat1} and~\ref{ln:single-sat2}).
Otherwise, it is added to $S_{\text{freed}}$ (Line~\ref{ln:single-unsat}).
By iteratively certifying freeable groups and isolating necessary singletons, the
procedure constructs $S_{\text{fixed}}$ as an explanation whose fixed values suffice
to preserve the prediction, while maximizing $S_{\text{freed}}$.

In incremental mode, each call to $\textsc{Verify}$ is issued with a query identifier
and a set of inherited identifiers $\mathcal{I}$, representing ancestor queries
whose learned conflicts may be reused.
As shown in Algorithm~\ref{alg:binary-sequential-inc}, inheritance is updated only
along branches where the current candidate set is not certified freeable, i.e., on
\textsf{SAT} and \textsf{TIMEOUT} outcomes.
Concretely, when the verification of a subset fails, the identifier of that query
is added to the inherited set passed to the recursive call
(Lines~\ref{ln:hi-sat} and~\ref{ln:lo-sat}).
In contrast, \textsf{UNSAT} outcomes certify freeability and do not induce further
recursion, and therefore do not require extending the ancestry set.

This design ensures that conflict inheritance is applied along
refinement branches, enabling effective reuse while preserving soundness.

\end{document}